\documentclass[11pt]{article}
\usepackage[margin=1in]{geometry}
\usepackage{amsmath,amssymb,amsthm,mathtools}
\usepackage[hidelinks]{hyperref}

\newtheorem{theorem}{Theorem}[section]
\newtheorem{lemma}[theorem]{Lemma}
\newtheorem{proposition}[theorem]{Proposition}
\newtheorem{corollary}[theorem]{Corollary}
\newtheorem{fact}[theorem]{Fact}
\theoremstyle{definition}

\theoremstyle{remark}
\newtheorem{remark}[theorem]{Remark}

\newcommand{\pP}{\mathord{\oplus}\mathrm{P}}
\newcommand{\BP}{\mathrm{BP}\cdot}
\newcommand{\BPP}{\mathrm{BPP}}
\newcommand{\NP}{\mathrm{NP}}
\newcommand{\PH}{\mathrm{PH}}
\newcommand{\cP}{\mathrm{P}}
\newcommand{\sP}{\#\mathrm{P}}
\newcommand{\Almost}{\mathrm{Almost}\text{-}}
\newcommand{\F}{\mathbb{F}_2}
\newcommand{\Maj}{\mathrm{Maj}}
\newcommand{\GF}{\mathrm{GF}}
\newcommand{\poly}{\mathrm{poly}}

\title{$\Almost\pP = \BP\pP$ and a Random-Oracle Proof of Toda's Theorem}
\author{Lance Fortnow}
\date{September 2026}

\begin{document}
\maketitle

\begin{abstract}
Using the recent exponential correlation bounds of Chattopadhyay, Hatami, Lee,
Lovett, Tal and Viola between $\F$-polynomials and the XOR of majorities, we
show that $\Almost\pP = \BP\pP = \BPP^{\pP}$, where $\Almost\pP$ is the class of
languages in $\pP^R$ for a random oracle $R$ with probability one. Combined with
the easy fact that $\PH^R \subseteq \pP^R$ relative to a random oracle, an idea
going back to Regan and Royer, this gives a proof of the first half of Toda's
theorem, $\PH \subseteq \BP\pP$, in which no probabilistic quantifier ever
has to be moved through an oracle. We compare the argument with the simple proof
of Toda's theorem in Fortnow (2009).
\end{abstract}

\section{Introduction}

For a relativizable class $\mathcal{C}$, let
\[
\Almost\mathcal{C} = \{\,L : \Pr_R[L \in \mathcal{C}^R] = 1\,\},
\]
where $R$ is chosen by putting each string in $R$ independently with
probability $1/2$. Bennett and Gill~\cite{BG81} and Ambos-Spies~\cite{AS86}
characterized $\BPP$ as $\Almost\cP$, and Nisan and Wigderson~\cite{NW94},
using a pseudorandom generator built from the hardness of parity for
constant-depth circuits, showed $\Almost\PH = \PH$. The natural analogue for
parity computation, $\Almost\pP = \BP\pP$, requires a pseudorandom generator
that fools polynomials over $\F$ of degree $\poly(n)$ on $2^{\poly(n)}$
variables with seed length $\poly(n)$. Earlier generators for low-degree
polynomials~\cite{Bog05,Lov09,Vio09} have seed length exponential in the
degree and are useless here. The breakthrough correlation bound of
Chattopadhyay, Hatami, Lee, Lovett, Tal and Viola~\cite{CHLLTV26} supplies
exactly the generator we need.

\begin{theorem}\label{thm:main}
$\Almost\pP = \BP\pP = \BPP^{\pP}$.
\end{theorem}

The motivation is Toda's theorem~\cite{Tod91}. Relative to a random oracle
$R$, the polynomial hierarchy collapses into $\pP^R$ by nothing more than
Valiant--Vazirani~\cite{VV86} and Papadimitriou--Zachos~\cite{PZ83} applied
level by level; this observation is due to Regan and Royer~\cite{RR95}, who
showed that $\PH^R$ is properly contained in $\pP^R$ for a random oracle $R$.
Since every language in $\PH$ is in $\PH^R$, Theorem~\ref{thm:main} then
yields $\PH \subseteq \BP\pP$.

\section{Preliminaries}

We identify an oracle $R$ with its characteristic sequence in $\{0,1\}^\omega$
and use the uniform (Lebesgue) measure $\mu$. For a finite string $\sigma$ let
$[\sigma]$ be the cylinder of oracles whose characteristic sequence begins
with $\sigma$. All oracle machines are clocked, so a machine running in time
$T(n)$ queries only strings of length at most $T(n)$.

A language $L$ is in $\pP^R$ if there is a polynomial-time nondeterministic
oracle machine $M$ such that $x\in L$ if and only if $M^R(x)$ has an odd
number of accepting paths; we write $M^R(x)\in\{0,1\}$ for this parity.
A language $L$ is in $\BP\pP$ if there are $A\in\pP$ and a polynomial $q$
with $\Pr_{r\in\{0,1\}^{q(|x|)}}[A(x,r)=L(x)]\ge 2/3$ for every $x$.

\begin{fact}[Papadimitriou--Zachos~\cite{PZ83}]\label{fact:pz}
For every oracle $A$, $\pP^{\pP^A} = \pP^A$, and hence $\cP^{\pP^A}=\pP^A$.
\end{fact}

In particular $\pP^A$ is closed under polynomial-time Turing reductions,
so $\BP\pP$ can be amplified to error $2^{-q(n)}$ for any polynomial $q$ by
taking majority votes, and $\BP\pP=\BPP^{\pP}$~\cite{Tod91}.

\begin{fact}[Valiant--Vazirani~\cite{VV86}]\label{fact:vv}
Let $W\subseteq\{0,1\}^m$ be nonempty. Choose $j\in\{1,\dots,m+1\}$ uniformly
and a random affine map $h(y)=Ay+b$ from $\F^m$ to $\F^j$. Then
$\Pr\big[\,|\{y\in W: h(y)=0\}| = 1\,\big]\ge 1/(8(m+1))$. The choice uses
$O(m^2)$ random bits.
\end{fact}

\section{The generator}

We need the generator of~\cite[Theorem A.1]{CHLLTV26}, whose analysis rests
on their main theorem: for every polynomial $p:\F^{k\ell}\to\F$ of degree at
most $d$,
\begin{equation}\label{eq:corr}
\Big|\,\mathbb{E}_x\big[(-1)^{p(x)+\bigoplus_{a=1}^k \Maj(x^{(a)})}\big]\Big|
\le \Big(\frac{2d}{\sqrt{\ell}}\Big)^k,
\end{equation}
where $x=(x^{(1)},\dots,x^{(k)})$ consists of $k$ blocks of $\ell$ bits and
$\ell$ is odd. The construction in~\cite{CHLLTV26} uses a combinatorial
design computed in time polynomial in the number of outputs. We have exponentially
many outputs, so we substitute a polynomial-based design in which each set can
be computed locally. The analysis is unchanged, and we include it for
completeness.

\begin{lemma}\label{lem:prg}
There is a polynomial $s$ and, for each $T\ge 1$, a map
$G_T:\{0,1\}^{s(T)}\to\{0,1\}^{\{0,1\}^{\le T}}$ such that
\begin{enumerate}
\item given $T$, $z\in\{0,1\}^{s(T)}$ and $w$ with $|w|\le T$, the bit
$G_T(z)_w$ can be computed in time $\poly(T)$, and
\item for every polynomial $p$ over $\F$ in the variables
$\{X_w:|w|\le T\}$ of degree at most $T$,
\[
\Big|\Pr_z[p(G_T(z))=1]-\Pr_{Y}[p(Y)=1]\Big|\le \frac{1}{10},
\]
where $Y$ is uniform on $\{0,1\}^{\{0,1\}^{\le T}}$.
\end{enumerate}
\end{lemma}

\begin{proof}
Let $N=2^{T+1}-1$ be the number of strings of length at most $T$, and set
$d=T$, $t=T+1$, $\varepsilon=1/10$, $k=\lceil\log(2N/\varepsilon)\rceil\le T+6$.
Let $\ell$ be the least odd integer with $\ell\ge 16d^2t^2$, and let $q$ be the
least power of two with $q\ge\ell$, so $q<2\ell=O(T^4)$.

\emph{Design.} Fix a set $A\subseteq\GF(q)$ of size $\ell$, ordered. For a
string $w$ with $|w|\le T$, write the string $1w$, padded on the left with
zeros to length $T+1$, as $(c_0,\dots,c_T)\in\{0,1\}^{T+1}\subseteq\GF(q)^{T+1}$
and let $f_w(\alpha)=\sum_{i} c_i\alpha^i$. Distinct $w$ give distinct
polynomials of degree at most $T$. Set
$S_w=\{(\alpha,f_w(\alpha)):\alpha\in A\}\subseteq\GF(q)^2$. Then $|S_w|=\ell$,
and for $w\ne v$ we have $|S_w\cap S_v|\le T<t$ since $f_w-f_v$ is a nonzero
polynomial of degree at most $T$.

\emph{Generator.} The seed consists of $k$ blocks
$z^{(1)},\dots,z^{(k)}\in\{0,1\}^{\GF(q)^2}$, so $s(T)=kq^2=O(T^9)$. Define
\[
G_T(z)_w=\bigoplus_{a=1}^k \Maj\big(z^{(a)}|_{S_w}\big).
\]
Computing $S_w$ takes $\poly(T)$ field operations, so item~1 holds.

\emph{Analysis.} Order the strings of length at most $T$ as
$w_1,\dots,w_N$, and let $Y_i$ consist of the first $i$ outputs of $G_T$
followed by $N-i$ independent uniform bits. Fix $p$ of degree at most $d$
and $i$. Fix all seed bits outside $S_{w_i}$ in every block, together with the
last $N-i$ bits. The remaining seed bits form uniform blocks
$x^{(1)},\dots,x^{(k)}\in\{0,1\}^\ell$, and output $i$ is
$F(x)=\bigoplus_a\Maj(x^{(a)})$. For $j<i$, output $j$ is an XOR over $a$ of
functions of $x^{(a)}|_{S_{w_j}\cap S_{w_i}}$, each depending on fewer than
$t$ variables, so it is a polynomial of degree less than $t$. For
$b\in\{0,1\}$, let $Q_b(x)$ be $p$ with these outputs, the bit $b$ in position
$i$, and the fixed later bits substituted; then $\deg Q_b< dt$. Viewing bits as
reals, the conditional difference between the acceptance probabilities under
$Y_i$ and $Y_{i-1}$ is
\[
\Big|\mathbb{E}_x\Big[Q_{F(x)}(x)-\tfrac{Q_0(x)+Q_1(x)}{2}\Big]\Big|
=\tfrac14\Big|\mathbb{E}_x\big[(-1)^{F(x)+Q_1(x)}-(-1)^{F(x)+Q_0(x)}\big]\Big|
\le \tfrac12\Big(\frac{2dt}{\sqrt\ell}\Big)^k\le 2^{-k-1}\le\frac{\varepsilon}{4N},
\]
by~\eqref{eq:corr} and $\sqrt\ell\ge 4dt$. Averaging over the fixed bits and
summing over the $N$ hybrids bounds the total difference by $\varepsilon/4$.
\end{proof}

\section{Almost-\texorpdfstring{$\oplus$}{Parity}P}

\begin{lemma}[Degree]\label{lem:degree}
Let $M$ be a nondeterministic oracle machine running in time $T(n)$. For every
$x$ there is a polynomial $p_x$ over $\F$ in the variables
$\{X_w:|w|\le T(|x|)\}$ of degree at most $T(|x|)$ such that $M^R(x)=p_x(R)$
for every oracle $R$.
\end{lemma}

\begin{proof}
Let $T=T(|x|)$. Fix a sequence $y\in\{0,1\}^T$ of nondeterministic choices.
The computation along $y$ is a decision tree of depth at most $T$ in the oracle
bits. The indicator of reaching a given accepting leaf is a product of at most
$T$ literals $X_w$ or $1+X_w$, and at most one leaf is reached, so the indicator
that path $y$ accepts is the $\F$-sum of these products. Summing over all $y$
gives $p_x$.
\end{proof}

\begin{lemma}[Uniformization]\label{lem:unif}
Let $M$ be a $\pP$ oracle machine with
$\mu\{R: L(M^R)=L\}>0$. Then there is a $\pP$ oracle machine $M'$ with
$\Pr_R[M'^R(x)=L(x)]\ge 0.99$ for every $x$.
\end{lemma}

\begin{proof}
Let $S=\{R:L(M^R)=L\}$, an intersection of clopen sets and hence measurable.
By outer regularity there is an open $O\supseteq S$ with
$\mu(O)\le\mu(S)/0.99$. Write $O$ as a disjoint union of cylinders
$[\sigma_i]$. Since $\sum_i\mu(S\cap[\sigma_i])=\mu(S)\ge 0.99\sum_i\mu([\sigma_i])$,
some $\sigma=\sigma_i$ has $\mu(S\cap[\sigma])\ge 0.99\,\mu([\sigma])$.

For an oracle $R$ let $R_\sigma$ agree with $\sigma$ on the first $|\sigma|$
strings and with $R$ elsewhere. Let $M'^R$ simulate $M^{R_\sigma}$, answering
queries among the first $|\sigma|$ strings from a hard-wired table. If $R$ is
uniform then $R_\sigma$ is uniform on $[\sigma]$, so
$\Pr_R[L(M'^R)=L]=\mu(S\cap[\sigma])/\mu([\sigma])\ge 0.99$, and in particular
$\Pr_R[M'^R(x)=L(x)]\ge 0.99$ for each $x$.
\end{proof}

\begin{proposition}\label{prop:hard}
$\Almost\pP\subseteq\BP\pP$.
\end{proposition}

\begin{proof}
Let $L\in\Almost\pP$. There are countably many $\pP$ oracle machines, so one
of them, $M$, has $\mu\{R:L(M^R)=L\}>0$. Let $M'$ be given by
Lemma~\ref{lem:unif}, with running time $T(n)=n^c$.

Fix $x$ of length $n$ and let $T=T(n)$. By Lemma~\ref{lem:degree},
$M'^R(x)=p_x(R)$ with $\deg p_x\le T$, and $p_x+L(x)+1$ also has degree at
most $T$. Since $R$ restricted to $\{0,1\}^{\le T}$ is uniform,
Lemma~\ref{lem:prg} gives
\[
\Pr_z\big[p_x(G_T(z))=L(x)\big]\ge \Pr_R\big[M'^R(x)=L(x)\big]-\tfrac1{10}\ge 0.89.
\]
Let $B=\{(x,z): p_x(G_{T(|x|)}(z))=1\}$. Then $B\in\pP$: simulate $M'$ on $x$,
and answer each query $w$ by computing $G_{T(|x|)}(z)_w$ in polynomial time.
Choosing $z$ uniformly from $\{0,1\}^{s(T(|x|))}$ puts $L$ in $\BP\pP$.
\end{proof}

\begin{proposition}\label{prop:easy}
$\BP\pP\subseteq\Almost\pP$.
\end{proposition}

\begin{proof}
Let $L\in\BP\pP$, amplified so that $A\in\pP$ and
$\Pr_r[A(x,r)\ne L(x)]\le 2^{-2|x|}$ with $|r|=q(|x|)$. Let $M^R$ on $x$ read
$r_i=R(\langle x,i\rangle)$ for $i\le q(|x|)$ deterministically and then run
the $\pP$ machine for $A(x,r)$. Then
$\sum_x\Pr_R[M^R(x)\ne L(x)]\le\sum_n 2^n2^{-2n}<\infty$, so by the
Borel--Cantelli lemma, with probability one $M^R$ errs on only finitely many
inputs. Patching those inputs with a finite table gives $L\in\pP^R$.
\end{proof}

Propositions~\ref{prop:hard} and~\ref{prop:easy}, together with
$\BP\pP=\BPP^{\pP}$, prove Theorem~\ref{thm:main}.

\begin{remark}
Only Proposition~\ref{prop:hard} is used below, and its proof uses nothing
from Toda's paper. The same argument relativizes: for every oracle $A$,
$\{L:\Pr_R[L\in\pP^{A\oplus R}]=1\}=\BP\pP^A$, since Lemma~\ref{lem:degree}
bounds the degree in the $R$-variables for each fixed $A$.
\end{remark}

\section{The polynomial hierarchy relative to a random oracle}

Let $\Sigma_0^R=\cP^R$ and $\Sigma_{k+1}^R=\NP^{\Sigma_k^R}$. Each language in
$\Sigma_k^R$ is given by a finite \emph{specification}: a tower
$\Pi=(N_k,\dots,N_1,N_0)$ of clocked oracle machines, with $N_0$ deterministic,
$N_{i}$ nondeterministic with oracle $L_{(N_{i-1},\dots,N_0)}(R)$ for
$i\ge 1$, and $N_0$ with oracle $R$. We write $L_\Pi(R)$ for the language
defined. Since $\Sigma_k^R$ is closed under join with $R$, we may assume every
$N_i$ has access to $R$ through its oracle. Composing the time bounds, there is
a polynomial $p_\Pi$ such that whether $x\in L_\Pi(R)$ depends only on $R$
restricted to strings of length at most $p_\Pi(|x|)$.

\begin{theorem}[cf.~Regan--Royer~\cite{RR95}]\label{thm:ph}
For every $k$ and every specification $\Pi$ at level $k$, we have
$\Pr_R[L_\Pi(R)\in\pP^R]=1$. Consequently,
$\Pr_R[\PH^R\subseteq\pP^R]=1$.
\end{theorem}

\begin{proof}
The consequence follows because there are countably many specifications.
We prove the first statement by induction on $k$. For $k=0$,
$\cP^R\subseteq\pP^R$ for every $R$.

Let $\Pi=(N,\Pi')$ be at level $k+1$, so $K^R=L_{\Pi'}(R)\in\Sigma_k^R$ and
$L_\Pi(R)=L(N^{K^R})$. Write $N$ as a deterministic verifier: $x\in L_\Pi(R)$
if and only if $W_x(R)\ne\emptyset$, where
$W_x(R)=\{y\in\{0,1\}^m: N^{K^R}(x,y)\text{ accepts}\}$ and $m=m(|x|)$ is a
polynomial. Let $p=p_\Pi$, so $W_x(R)$ depends only on $R$ restricted to
strings of length at most $p(|x|)$.

\emph{Hashing from far away.} Let $s(n)=16(m+1)n$. For $i\le s(|x|)$, read
the random choices $(j_i,h_i)$ of Fact~\ref{fact:vv} from the bits of $R$ at
the strings $1^{p(|x|)+1}0\langle x,i,b\rangle$, $b=1,2,\dots$; all of these
have length greater than $p(|x|)$. Define
\[
V^R(x)=\bigvee_{i\le s(|x|)}\Big[\,\big|\{y\in W_x(R):h_i(y)=0\}\big|\text{ is odd}\,\Big].
\]
For every $R$ the language $V^R$ is in $\cP^{\pP^{K^R\oplus R}}$: each
bracket is a $\pP^{K^R\oplus R}$ query and a polynomial-time machine takes
the OR.

\emph{Error.} If $W_x(R)=\emptyset$ then $V^R(x)=0$. Condition on $R$
restricted to strings of length at most $p(|x|)$; this fixes $W_x(R)$ and
leaves the hash bits uniform and independent. If $W_x(R)\ne\emptyset$, each
$i$ isolates a unique element with probability at least $1/(8(m+1))$, so
\[
\Pr_R[V^R(x)\ne L_\Pi(R)(x)]\le\Big(1-\frac{1}{8(m+1)}\Big)^{s(|x|)}\le
e^{-2|x|}\le 2^{-2|x|}.
\]
By the Borel--Cantelli lemma, with probability one $V^R$ and $L_\Pi(R)$
differ on only finitely many inputs.

\emph{Collapse.} By induction, with probability one $K^R\in\pP^R$. For
every $R$ in both probability-one events, Fact~\ref{fact:pz} gives
$V^R\in\cP^{\pP^{\pP^R}}=\pP^R$, and $L_\Pi(R)$ is a finite variation of
$V^R$, so $L_\Pi(R)\in\pP^R$.
\end{proof}

\begin{remark}
Reading the hash functions from strings longer than $p(|x|)$ is what makes the
argument work: an adversarial $\Sigma_{k+1}^R$ machine may query any string it
can reach, and the Valiant--Vazirani analysis needs the hash functions to be
independent of the witness set.
\end{remark}

\section{Toda's theorem}

\begin{corollary}\label{cor:toda}
$\PH\subseteq\BP\pP$, and hence $\PH\subseteq\cP^{\sP}$.
\end{corollary}

\begin{proof}
If $L\in\PH$ then $L\in\Sigma_k\subseteq\Sigma_k^R$ for some $k$ and every
$R$, so by Theorem~\ref{thm:ph}, $L\in\pP^R$ with probability one. Thus
$L\in\Almost\pP$, and Proposition~\ref{prop:hard} gives $L\in\BP\pP$. The
second statement follows from Toda's inclusion $\BP\pP\subseteq\cP^{\sP}$
\cite{Tod91}; see~\cite{For09} for a short proof of that step.
\end{proof}

\section{Discussion}

\paragraph{Regan and Royer.} Regan and Royer~\cite{RR95} simplified Toda's
proof that $\PH\subseteq\BP\pP$ and showed that $\PH^R$ is properly contained
in $\pP^R$ relative to a random oracle $R$. Theorem~\ref{thm:ph} is the
containment half of their result. Going from the random-oracle statement back
to Toda's unrelativized theorem requires exactly
$\Almost\pP\subseteq\BP\pP$, which is where the correlation bounds
of~\cite{CHLLTV26} enter.

\paragraph{The 2009 simple proof.} Toda's original argument shows
$\PH\subseteq\BP\pP$ by induction on the levels. Its main difficulty is
showing $\BP\pP^{\BP\pP}\subseteq\BP\pP$: the random bits inside an oracle
must be amplified, chosen in advance, and pulled outside. The proof
in~\cite{For09} avoids that step. Valiant--Vazirani gives
$\NP\subseteq\BPP^{\pP}$; relativizing to $\pP$ and applying
Papadimitriou--Zachos gives $\NP^{\pP}\subseteq\BPP^{\pP}$; and the
relativized form of Zachos's theorem~\cite{Zac88}, that
$\NP^A\subseteq\BPP^A$ implies $\PH^A\subseteq\BPP^A$, applied with $A=\pP$
gives $\PH\subseteq\PH^{\pP}\subseteq\BPP^{\pP}$. The randomness bookkeeping is
isolated inside Zachos's theorem, whose proof amplifies and chooses the random
bits before the nondeterministic guesses.

The random-oracle proof also avoids moving probabilistic quantifiers through
oracles, but in a different way. The randomness is global: once $R$ is fixed,
every level of the induction is an ordinary relativized inclusion of classes,
and the probability enters only through a Borel--Cantelli argument about
finitely many errors. Nothing is ever amplified inside an oracle. The price is
paid once, at the end: an exponentially long random oracle has to be replaced
by polynomially many random bits, and a $\pP$ computation looks at the whole
oracle through the parity of exponentially many paths. That replacement needs a
generator fooling $\F$-polynomials of degree $\poly(n)$ on $2^{\poly(n)}$
variables with seed length $\poly(n)$, equivalently an explicit function with
exponentially small correlation with polynomials of polynomial degree. Before
\cite{CHLLTV26}, no such bound was known even for degree $\log n$.

So the two proofs make opposite trades. The 2009 proof is elementary and
self-contained. The random-oracle proof reduces the collapse of the hierarchy
to a triviality and puts all of the difficulty into a single
hardness-versus-randomness step, in the same way that Nisan and
Wigderson~\cite{NW94} obtained $\Almost\PH=\PH$ from the hardness of parity for
constant-depth circuits. It needs a deep circuit lower bound to recover a
theorem that has an elementary proof, but it explains the first half of Toda's
theorem as a collapse that holds relative to a random oracle, plus a
derandomization of that oracle.

\paragraph{Acknowledgment.} This note was drafted with the assistance of
Claude (Anthropic).

\end{document}